\documentclass[a4paper,UKenglish]{lipics-v2018}

\usepackage{microtype}%

\title{Reachability and Distances under Multiple Changes}%

\author{Samir Datta}{Chennai Mathematical Institute \& UMI ReLaX, Chennai, India}{sdatta@cmi.ac.in}{}{}
\author{Anish Mukherjee}{Chennai Mathematical Institute, Chennai, India}{anish@cmi.ac.in}{}{}
\author{Nils Vortmeier}{TU Dortmund University, Dortmund, Germany}{nils.vortmeier@tu-dortmund.de}{}{}
\author{Thomas Zeume}{TU Dortmund University, Dortmund, Germany}{thomas.zeume@tu-dortmund.de}{}{}

\authorrunning{S. Datta, A. Mukherjee, N. Vortmeier, T. Zeume} %
\Copyright{Samir Datta, Anish Mukherjee, Nils Vortmeier and Thomas Zeume}%
\subjclass{Theory of computation $\rightarrow$ Models of computation,
Theory of computation $\rightarrow$ Finite Model Theory}
 

\relatedversion{This is the full version of a paper to be presented at ICALP 2018.}

\funding{The authors acknowledge the financial support by the DAAD-DST grant ``Exploration of New Frontiers in Dynamic Complexity''. The first and the second authors were partially funded by a grant from Infosys foundation. The second author was partially supported by a TCS PhD fellowship. 
The last two authors acknowledge the financial support by DFG grant SCHW 678/6-2 on ``Dynamic Expressiveness of Logics''.}

\acknowledgements{We thank Rajit Datta and Partha Mukhopadhyay for pointing us to the Agrawal-Vinay technique for computing determinants. We also thank Thomas Schwentick for many illuminating discussions on dynamic complexity.}

\EventEditors{Ioannis Chatzigiannakis, Christos Kaklamanis, D\'{a}niel Marx, and Don Sannella}
\EventNoEds{4}
\EventLongTitle{45th International Colloquium on Automata, Languages, and Programming (ICALP 2018)}
\EventShortTitle{ICALP 2018}
\EventAcronym{ICALP}
\EventYear{2018}
\EventDate{July 9--13, 2018}
\EventLocation{Prague, Czech Republic}
\EventLogo{eatcs}
\SeriesVolume{107}
\ArticleNo{165}

\usepackage{mathtools}

\usepackage{xspace}
\usepackage{ifmtarg}
\usepackage{tikz}
\usepackage{algorithm}
\usepackage[noend]{algpseudocode}
\algnewcommand\algorithmicinput{\textbf{Input:}}
\algnewcommand\Input{\item[\algorithmicinput]}
\algnewcommand\algorithmicoutput{\textbf{Output:}}
\algnewcommand\Output{\item[\algorithmicoutput]}
\usepackage{varwidth}

\usepackage{thmtools}
\usepackage{thm-restate}
\declaretheoremstyle[
    spaceabove=6pt, 
    spacebelow=6pt, 
    headfont=\sffamily\bfseries, 
    bodyfont = \normalfont\itshape,
    postheadspace=2mm, 
    headpunct={.}]{theoremstyle}

\makeatletter
\def\BState{\State\hskip-\ALG@thistlm}
\makeatother

\usetikzlibrary{arrows,decorations.pathmorphing,backgrounds,positioning,fit,petri,backgrounds}
\usetikzlibrary{calc}
\usetikzlibrary{scopes}
\usetikzlibrary{patterns}
\usepgflibrary{shapes.geometric}
\usetikzlibrary{matrix}
\usetikzlibrary{decorations.pathreplacing}

\nolinenumbers
\newif\ifcomments
\newif\ifchanges
\commentsfalse\changesfalse
\newcommand{\onlyLong}[1]{#1}

\hideLIPIcs

\newcommand  {\myclass} [1]  {\ensuremath{\textsc{#1}}}

\newcommand{\StaClass}[1]{\myclass{#1}\xspace}

\newcommand{\DynClass}[1]{\myclass{Dyn#1}\xspace}

\newcommand  {\myproblem} [1] {\textsc{#1}}

\newcommand  {\problem}[1] {\myproblem{#1}}

\newcommand     {\LOGSPACE}     {\myclass{LOGSPACE}}

\newcommand     {\NL}   {\myclass{NL}}

\newcommand     {\NC}   {\myclass{NC}}
\newcommand     {\AC}   {\myclass{AC}}
\newcommand     {\TC}   {\myclass{TC}}

\newcommand{\FO}{\StaClass{FO}}
\newcommand{\FOMaj}{\StaClass{FO+Maj}}
\newcommand{\FOMajar}{\StaClass{FO+Maj$(+,\!\times\!)$}}
\newcommand{\FOar}{\StaClass{FO$(+,\!\times\!)$}}
\newcommand{\MSO}[1][\quant]{\StaClass{MSO}}

\newcommand{\CQ}[1][]{\StaClass{CQ}}
\newcommand{\UCQ}[1][]{\StaClass{UCQ}}
\newcommand{\CQneg}[1][]{\StaClass{CQ\ensuremath{^{\mneg}}}}
\newcommand{\UCQneg}[1][]{\StaClass{UCQ\ensuremath{^{\mneg}}}}

\newcommand{\mneg}{\neg} %

\newcommand{\DynFO}{\DynClass{FO}}
\newcommand{\DynNC}{\DynClass{NC}}
\newcommand{\DynFOMaj}{\DynClass{FO+Maj}}

\newcommand{\DynFOar}{\DynClass{FO$(+,\times)$}}

\newcommand{\Emptiness}[1][]{\problem{Emptiness}\ifthenelse{\equal{#1}{}}{}{(#1)}\xspace}
\newcommand{\Consistency}[1][]{\problem{Consistency}\ifthenelse{\equal{#1}{}}{}{(#1)}\xspace}
\newcommand{\HI}[1][]{\problem{HistoryIndependence}\ifthenelse{\equal{#1}{}}{}{(#1)}\xspace}

\newcommand{\mtext}[1]{\textsc{#1}}

\providecommand {\calA}      {{\mathcal A}\xspace}

\providecommand {\calC}      {{\mathcal C}\xspace}
\providecommand {\calD}      {{\mathcal D}\xspace}

\providecommand {\calI}      {{\mathcal I}\xspace}

\providecommand {\calP}      {{\mathcal P}\xspace}
\providecommand {\calQ}      {{\mathcal Q}\xspace}

\providecommand {\calS}      {{\mathcal S}\xspace}

\newcommand{\N}{\ensuremath{\mathbb{N}}}

\newcommand{\Z}{\ensuremath{\mathbb{Z}}}

\newcommand{\Q}{\ensuremath{\mathbb{Q}}}

\newcommand{\R}{\ensuremath{\mathbb{R}}}

\newcommand{\F}{\ensuremath{\mathbb{F}}}

\newcommand{\bigO}{\ensuremath{\mathcal{O}}}
\newcommand{\tpl}{\bar}

\newcommand{\restrict}[2]{#1\mspace{-3mu}\upharpoonright \mspace{-3mu}#2}

\newcommand{\df}{\ensuremath{\mathrel{\smash{\stackrel{\scriptscriptstyle{
    \text{def}}}{=}}}} \;}

\makeatletter %
\newcommand{\auxramsey}[4]{
  \@ifmtarg{#1}{
    \@ifmtarg{#4}{
      \ensuremath{R(#2; #3)}
    }{
      \ensuremath{R^#4(#2; #3)}
    }
   }{
    \@ifmtarg{#4}{
      \ensuremath{R_{#1}(#2; #3)}
    }{
      \ensuremath{R^#4_{#1}(#2; #3)}
    }
  }
}

   \theoremstyle{plain}

   \newtheorem{proposition}[theorem]{Proposition}
   
   \newtheorem*{proposition*}{Proposition}

    \theoremstyle{definition}
    \newenvironment{proofof}[1]{\begin{proof}[Proof (of #1).]}{\end{proof}}
    \newenvironment{proofsketchof}[1]{\begin{proof}[Proof sketch (of #1).]}{\end{proof}}
\newcommand{\arity}{\ensuremath{\text{Ar}}}

\newcommand{\schema}{\ensuremath{\tau}\xspace}

\newcommand{\struc}{\calS}

\newcommand{\quant}{\mathbb{Q}}

\newcommand{\db}{\ensuremath{\calD}\xspace}
\newcommand{\inp}{\ensuremath{\calI}\xspace}
\newcommand{\aux}{\ensuremath{\calA}\xspace}

\newcommand{\domain}{\ensuremath{ D}\xspace}

\newcommand{\query}{\ensuremath{\calQ}}

\newcommand{\adom}{\ensuremath{\text{adom}}}

\newcommand{\state}{\ensuremath{\struc}\xspace}

\newcommand{\prog}{\ensuremath{\calP}\xspace}

\newcommand{\uf}[4]{
  \@ifmtarg{#4}{
    \ensuremath{\phi^{#1}_{#2}(#3)}
   }{
    \ensuremath{\phi^{#1}_{#2}(#3; #4)}
  }
}
\newcommand{\huf}[4]{
  \@ifmtarg{#4}{
    \ensuremath{\widehat{\phi}^{#1}_{#2}(#3)}
   }{
    \ensuremath{\widehat{\phi}^{#1}_{#2}(#3; #4)}
  }
}

\newcommand{\ufb}[4]{
  \@ifmtarg{#4}{
    \ensuremath{\psi^{#1}_{#2}(#3)}
   }{
    \ensuremath{\psi^{#1}_{#2}(#3; #4)}
  }
}

  \makeatletter %
  \newcommand{\ufsubstitute}[5]{
    \@ifmtarg{#5}{
      \ensuremath{\phi^{#2}_{#3}[#1](#4)}
    }{
      \ensuremath{\phi^{#2}_{#3}[#1](#4; #5)}
    }
  }

  \makeatletter %
  
  \makeatletter %

\newcommand{\ut}[4]{
  \@ifmtarg{#4}{
    \ensuremath{t^{#1}_{#2}(#3)}
   }{
    \ensuremath{t^{#1}_{#2}(#3; #4)}
  }
}

\newcommand{\ite}[3]{
  \@ifmtarg{#1}{
    \ensuremath{\mtext{ITE}}
   }{
    \mtext{ITE}(#1,#2,#3)  
  }
}

\newcommand{\mf}[3]{
  \@ifmtarg{#3}{
    \ensuremath{\mu_{#1}(#2)}
   }{
    \ensuremath{\mu_{#1}(#2; #3)}
  }
}

\newcommand{\mfos}[4]{
  \@ifmtarg{#4}{
    \ensuremath{{#1}_{#2}(#3)}
   }{
    \ensuremath{{#1}_{#2}(#3; #4)}
  }
}

\providecommand{\nc}{\newcommand}

\usepackage{mdframed}

\ifcomments
\makeatletter
\newmdenv[
  innertopmargin=1mm,
  innerbottommargin=1mm,
]{fcomment@inner}
\newenvironment{fcomment}[1][]
  {%
    \if\relax\detokenize{#1}\relax
      \def\fcomment@name{}%
    \else
      \def\fcomment@name{#1}%
    \fi
    \fcomment@inner \textbf{\fcomment@name}:
  }
  {\endfcomment@inner}
  
\nc{\commentbox}[1]{\noindent\framebox{\parbox{0.98\linewidth}{#1}}}
\nc{\todo}[1]{\ \\ {\color{red} \fbox{\parbox{0.98\linewidth}{{\sc
          ToDo}:\\  #1}}}}

\newcommand{\acomment}[2]{\begin{fcomment}[#1]#2\end{fcomment}}
\newcommand{\mcomment}[2]{{\color{blue}(#1)}\footnote{#1: #2}} %
\else
\nc{\commentbox}[1]{}
\newcommand{\mcomment}[2]{}
\newcommand{\acomment}[2]{}
\fi

\ifchanges

\newcommand{\loldnew}[3]{\commentbox{{\textcolor{blue}{\setlength{\fboxsep}{1pt}\fbox{\small
          #1}}} \textcolor{red}{\footnotesize #2}}
  \textcolor{blue}{#3}}
\setul{}{0.2mm}
\setstcolor{red}
\newcommand{\oldnew}[3]{{\textcolor{blue}{\setlength{\fboxsep}{1pt}\fbox{\small
        #1}}} \st{\footnotesize #2} \textcolor{blue}{#3}}

\else
\newcommand{\loldnew}[3]{#3}
\newcommand{\oldnew}[3]{#3}
\fi

 \nc{\tzm}[1]{\mcomment{TZ}{#1}}
 \nc{\tsm}[1]{\mcomment{TS}{#1}}
 \nc{\nilsm}[1]{\mcomment{NV}{#1}}
 \nc{\sdm}[1]{\mcomment{SD}{#1}}
 \nc{\amm}[1]{\mcomment{AM}{#1}}

 \nc{\tz}[1]{\acomment{TZ}{#1}}
 \nc{\thz}[1]{\acomment{TZ}{#1}}
 \nc{\ts}[1]{\acomment{TS}{#1}}
 \nc{\nils}[1]{\acomment{NV}{#1}}
 \nc{\sd}[1]{\acomment{SD}{#1}}
 \nc{\am}[1]{\acomment{AM}{#1}}  

\nc{\tzon}[2][]{\oldnew{TZ}{#1}{#2}} 
\nc{\tson}[2][]{\oldnew{TS}{#1}{#2}}
\nc{\nilson}[2][]{\oldnew{NV}{#1}{#2}}
\nc{\sdon}[2][]{\oldnew{SD}{#1}{#2}}
\nc{\amon}[2][]{\oldnew{AM}{#1}{#2}}

\nc{\tzlon}[2][]{\loldnew{TZ}{#1}{#2}} 
\nc{\tslon}[2][]{\loldnew{TS}{#1}{#2}}
\nc{\nilslon}[2][]{\loldnew{NV}{#1}{#2}}
\nc{\sdlon}[2][]{\oldnew{SD}{#1}{#2}}
\nc{\amlon}[2][]{\oldnew{AM}{#1}{#2}}

\begin{document}
  \maketitle
 \begin{abstract}
   Recently it was shown that the transitive closure of a directed graph can be updated using first-order formulas after insertions and deletions of single edges in the dynamic descriptive complexity framework by Dong, Su, and Topor, and Patnaik and Immerman. In other words, Reachability is in \DynFO.
   
   In this article we extend the framework to changes of multiple edges at a time, and study the Reachability and Distance queries under these changes. We show that the former problem can be maintained in~$\DynFO(+, \times)$ under changes affecting~$\bigO(\frac{\log n}{\log \log n})$ nodes, for graphs with $n$ nodes. If the update formulas may use a majority quantifier then both Reachability and Distance can be maintained under changes that affect $\bigO(\log^c n)$ nodes, for fixed $c \in \N$. Some preliminary results towards showing that distances are in $\DynFO$ are discussed.
 \end{abstract}

\section{Introduction}\label{section:intro}
In today's databases, data sets are often large and subject to frequent changes. 
In use cases where only a fixed set of queries has to be evaluated on such data, it is not efficient to re-evaluate queries after each change, and therefore dynamic approaches have been considered. The idea is that when a database $\calD$ is modified by changing a set  $\Delta \calD$ of tuples then the result of a query is recomputed by using its result on $\calD$, the set $\Delta \calD$, and possibly other previously computed auxiliary data. 

One such dynamic approach is the \emph{dynamic descriptive complexity} approach, formulated independently by Dong, Su, and Topor \cite{DongST95}, as well as Patnaik and Immerman \cite{PatnaikI97}. In their framework the query result and the auxiliary data are represented by relations, and updates of the auxiliary relations are performed by evaluating first-order formulas. The class of queries that can be maintained in this fashion constitutes the class $\DynFO$. The motivation to use first-order logic as the vehicle for updates is that its evaluation is highly parallelizable and, in addition, that it corresponds to the relational algebra which is the core of SQL. Hence, if a query result can be maintained using a first-order update program, this program can be translated into equivalent SQL queries. 

While it is desirable to understand how to update query results under complex changes~$\Delta \calD$, the focus of dynamic descriptive complexity so far has been on single tuple changes. The reason is that for many queries our techniques did not even suffice to tackle this case. 

In recent years, however, we have seen several new techniques for maintaining queries. The Reachability query --- one of the main objects of study in dynamic descriptive complexity --- has been shown to be in $\DynFO$ using a linear algebraic method and a simulation technique~\cite{DattaKMSZ15}. The latter has been advanced into a very powerful tool: for showing that a query can be maintained in $\DynFO$, it essentially suffices to show that it can be maintained for $\log n$ many change steps after initializing the auxiliary data by an $\AC^1$ pre-computation\footnote{Readers not familiar with the circuit class $\AC^1$ may safely think of $\LOGSPACE$ pre-computations.}~\cite{Datta0SVZ17}, where $n$ is the size of the database's (active) domain.   This tool has been successfully applied to show that all queries expressible in monadic second order logic can be maintained in $\DynFO$ on structures of bounded treewidth. 

Those new techniques motivate a new attack on more complex changes $\Delta \calD$. But what are reasonable changes to look at? Updating a query after a change $\Delta \calD$ that replaces the whole database by a new database is essentially equivalent to the static evaluation problem with built-in relations: the stored auxiliary data has to be helpful for every possible new database, and therefore plays the role of built-in relations. Thus changes should be restricted in some way. Three approaches come to mind immediately: to only allow changes of restricted size; to restrict changes structurally; or to define changes in a declarative way. 

In this article we focus on the first approach. Before discussing our results we shortly outline the other two approaches.

There is a wide variety of structural restrictions. For example, the change set $\Delta \calD$ could only change the database locally or in such a way that the changes affect auxiliary relations only locally, e.g., if edges are inserted into distinct connected components it should be easier to maintain reachability. Another option is to restrict $\Delta \calD$ to be of a certain shape, examples studied in the literature are cartesian-closed changes \cite{DongST95} and deletions of anti-chains \cite{DongP97}. 

A declarative mechanism for changing a database is to provide a set of parameterised rules that state which tuples should be changed depending on a parameter provided by a user. For example, a rule $\rho(x,y; z)$ could state that all edges $(x,y)$ shall be inserted into a graph such that $x$ and $y$ are connected to the parameter $z$. First-order logic as a declarative mean to change databases has been studied in \cite{SchwentickVZ17}, where it was shown that undirected reachability can be maintained under insertions defined by first-order formulas, and single tuple deletions.

In this article we study changes of small size with a focus on the Reachability and Distance queries. As can be seen from the discussion above, the former query has been well-studied in diverse settings of dynamic descriptive complexity, and therefore results on its maintainability under small changes serve as an important reference point.

There is another reason to study Reachability under non-constant size changes. Recall that Reachability is complete for the static complexity class $\NL$. 
The result that Reachability is in $\DynFO$ does not imply $\NL \subseteq \DynFO$, as $\DynFO$ is only known to be closed under very weak reductions, called bounded first-order reductions, under which Reachability is not $\NL$-complete \cite{PatnaikI97}.
In short, these reductions demand that whenever a bit of an instance is changed, then only constantly many bits change in the image of the instance under the reduction.
When a query such as Reachability is maintainable under larger changes, then this restriction may be relaxed and might yield new maintainability results for other queries under single edge changes.

In this work we show that Reachability can be maintained under changes of non-constant size. Since our main interest is the study of changes of non-constant size, we assume throughout the article that all classes come with built-in arithmetic and denote, e.g., by $\DynFOar$ the class of queries that can be maintained with first-order updates in the presence of a built-in linear addition and multiplication relations. How our results can be adapted to classes without built-in arithmetic is discussed towards the end of Section \ref{section:framework}. 
\begin{restatable}{rtheorem}{reachabilitylognloglogn}
\label{theorem:reachability_multiple_changes}
Reachability can be maintained in $\DynFOar$ under changes that affect $\bigO(\frac{\log n}{\log \log n})$ nodes of a graph, where $n$ is the number of nodes of the graph.
\end{restatable}

The distance query was shown to be in $\DynFOMaj$ by Hesse \cite{Hesse03}, where the class $\DynFOMaj$ allows to specify updates with first-order formulas that may include majority quantifiers (equivalently, updates can be specified by uniform $\TC^0$ computations). We generalize Hesse's result to changes of size polylogarithmic in the size of the domain.

\begin{restatable}{rtheorem}{distancespolylogn}
\label{theorem:distances_multiple_changes}
  Reachability and Distance can be maintained in %
  $\DynFOMaj(+, \times)$ under changes that affect $\bigO(\log^c n)$ nodes of a graph, where $c \in \N$ is fixed and $n$ is the number of nodes of the graph. 
\end{restatable}

One of the important open questions of dynamic descriptive complexity is whether distances can be maintained in $\DynFO$, even under single edge changes. 
We contribute to the solution of this question by discussing how distances can be maintained in a subclass of $\DynFOMaj(+, \times)$ that is only slightly stronger than $\DynFO(+, \times)$.

\paragraph*{Organization} 
After recapitulating notations in Section \ref{section:preliminaries}, we adapt the dynamic complexity framework to bulk changes in Section \ref{section:framework}. Our main results, maintainability of reachability and distances under multiple changes, are proved in Section \ref{section:reachability_woodbury} and Section \ref{section:distances_woodbury}. We conclude with a discussion in Section \ref{section:conclusion}.

\section{Preliminaries}\label{section:preliminaries}
In this section we review basic definitions and results from finite model theory and databases.

We consider finite relational structures over relational signatures $\schema= \{R_1, \ldots, R_\ell\}$, where each $R_i$ is a relational symbol of arity $\arity(R_i)$. A $\schema$-structure $\db$ consists of a finite domain $\domain$ and relations $R_i^{\db}$ over $\domain$ of arity  $\arity(R_i)$, for each $i \in \{1, \ldots, \ell\}$.
The \emph{active domain} $\adom(\db)$  of a structure $\db$ contains all elements used in some tuple of $\db$. 
Since the motivation to study dynamic complexity originates from database theory, we use terminology from this area. In particular we use the terms ``relational structure'' and ``relational database'' synonymously.

We study the queries Reachability and Distance.
Reachability asks, given a directed graph $G$, for all pairs $s, t$ of nodes such that there is a path from $s$ to $t$ in $G$.
Distance asks for the length of the shortest path between any pair of reachable nodes.

We assume familiarity with first-order logic \FO and refer to \cite{ImmermanDC} for an introduction. 
The logic \FOMaj extends \FO by allowing majority quantifiers. Such quantifiers can ask whether more than half of all elements satisfy a given formula.
We write \FOar and \FOMajar to denote that formulas have access to built-in relations $\leq,+,\times$ which are interpreted as linear order, addition and multiplication on the domain of the underlying structure.
We note that \FOar and \FOMajar are equal to the circuit classes (\StaClass{DLOGTIME}-)uniform $\AC^0$ and $\TC^0$, respectively \cite{BarringtonIS90}.

In $\FO(+, \times)$, each tuple $(a_1, \ldots, a_c)$ encodes a number from $[n^{c}-1]_0 \df \{0, \ldots n^c-1\}$. %
We will henceforth identify tuples over the domain and numbers. 

It is well-known that $\FO(+, \times)$ supports arithmetic on numbers with polylog bits. Furthermore, iterated addition and multiplication for polylog many numbers with polylog bits can be expressed in $\FO(+, \times)$. More precisely:

\begin{lemma}[{cf.~\cite[Theorem 5.1]{HesseAB02}}]\label{lemma:computations_in_fo}
  Suppose $\varphi$ is a $\FO(+, \times)$ formula that defines $r \in \bigO(\log^c n)$ polylog bit numbers $a_1, \ldots, a_r$,  then there are formulas $\psi_+$ and $\psi_\times$ that define the sum and product of $a_1, \ldots, a_r$, respectively.
\end{lemma}

Due to these facts, many calculations can be defined in $\FO(+, \times)$. In particular, primes can be identified, and $\frac{\log n}{\log \log n}$ numbers of $\log \log n$ bits each can be encoded and decoded in $\log n$ bit numbers.

Suppose $p_1, \ldots, p_m$ are primes whose product is $N$. Then each number $A < N$ can be uniquely represented as a tuple $\tpl a = (a_1, \ldots, a_m)$ where $a_i = A \mod p_i$. The tuple $\tpl a$ is called \emph{Chinese remainder representation} (CRR) of $A$. The number $A$ can recovered from $\tpl a$ via $A = \sum_i a_i h_i C_i - r N$, where $C_i = \frac{N}{m_i}$, $h_i$ is the inverse of $C_i$ modulo $m_i$, and $r= \sum_{i=1}^m \lfloor \frac{x_i h_i}{m_i} \rfloor$ \cite[p. 702]{HesseAB02}.
Due to Lemma \ref{lemma:computations_in_fo}, in $\FO(+,\times)$ one can encode and decode $\bigO(\log n)$ bit numbers into their CRR defined by $\bigO(\log n)$ primes with $\bigO(\log \log n)$ bits.

In this article we use basic notions and results from linear algebra which are introduced when they are needed. 
Throughout the article, a matrix with $\bigO(n^d)$ rows and columns and entries in $[n^c]_0$ will be represented by a relation $R$ that contains a tuple $(\bar r, \bar c, \bar v)$ if and only if the value at row $\bar r$ and column $\bar c$ is $\bar v$.
 
\section{Dynamic Framework for Multiple Changes}\label{section:framework}
We briefly repeat the essentials of dynamic complexity, closely following \cite{SchwentickZ16}, and discuss generalisations due to changes of non-constant size. 

The goal of a dynamic program is to answer a given query on an \emph{input database} subjected to changes that insert or delete tuples. The program may use an auxiliary data structure represented by an \emph{auxiliary database} over the same domain. Initially, both input and auxiliary database are empty; and the domain is fixed during each run of the program. 

\subparagraph*{Changes}
In previous work, changes of single tuples have been represented as explicit parameters for the formulas used to update the auxiliary relations. 
Non-constant size changes cannot be represented in this fashion. An alternative is to represent changes implicitly by giving  update formulas access to the old input database as well as to the changed input database \cite{GradelS12}. Here, we opt for this approach. %

For a database $\db$ over domain $\domain$ and schema $\schema$, a change $\Delta \db$ consists of sets $R^{+}$ and $R^{-}$ of tuples for each relation symbol $R \in \schema$. The result $\db + \Delta \db$ of an application of the change $\Delta \db$  to $\db$ is the input database where $R^{\db}$ is changed to $(R^{\db} \cup R^{+}) \setminus R^{-}$. The \emph{size} of $\Delta \db$ is the total number of tuples in relations $R^{+}$ and $R^{-}$ and the set of \emph{affected elements} is the (active) domain of tuples in $\Delta \db$.

\subparagraph*{Dynamic Programs and Maintenance of Queries} A dynamic program consists of a set of update rules that specify how auxiliary relations are updated after changing the input database. An \emph{update rule} for updating an $\ell$-ary auxiliary relation $T$ after a change is a first-order formula $\varphi$ over schema $\tau \cup \tau_{\text{aux}}$ with $\ell$ free variables, where $\tau_{\text{aux}}$ is the schema of the auxiliary database. After a change $\Delta \db$, the new version of $T$ is $T \df \{ \vec a \mid (\db + \Delta \db, \aux) \models \varphi(\vec a)\}$ where $\db$ is the old input database and $\aux$ is the current auxiliary database. Note that a dynamic program can choose to have access to the old input database by storing it in its auxiliary relations. %

For a state $\state = (\db, \aux)$ of the dynamic program $\prog$ with input database $\db$ and auxiliary database $\aux$ we denote the state of the program after applying a change sequence $\alpha$ and updating the auxiliary relations accordingly by $\prog_\alpha(\state)$. 

The dynamic program \emph{maintains} a $q$-ary query $\query$ under changes that affect $k$ elements (under changes of size $k$, respectively)  if it has a $q$-ary auxiliary relation $Q$ that at each point stores the result of $\query$ applied to the current input database. More precisely, for each non-empty sequence $\alpha$ of changes that affect $k$ elements (changes of size $k$, respectively), the relation $Q$ in $\prog_\alpha(\state_\emptyset)$ and $\query(\alpha(\db_\emptyset))$ coincide, where $\db_\emptyset$ is an empty input structure, $\state_\emptyset$ is the auxiliary database with empty auxiliary relations over the domain of $\db_\emptyset$, and $\alpha(\db_\emptyset)$ is the input database after applying $\alpha$. 

If a dynamic program maintains a query, we say that the query is in $\DynFO$.
Similarly to \DynFO one can define the class of queries \DynFOar that allows for three particular auxiliary relations that are initialised as a linear order and the corresponding addition and multiplication relations. Other classes are defined accordingly.

For many natural queries $\query$, in order to show that $\query$ can be maintained, it is enough to show that the query can be maintained for a bounded number of steps. Intuitively, this is possible for queries for which isolated elements do not influence the query result, if there are many such elements. Formally, a query $\query$ is \emph{almost domain-independent} if there is a $c \in \N$ such that $\restrict{\query(\calA)}{(\adom(\calA) \cup B)} = \query(\restrict{\calA}{(\adom(\calA) \cup B)})$ for all structures $\calA$ and sets $B \subseteq A \setminus \adom(\calA)$ with $|B| \geq c$.

A query $\query$ is \emph{$(\calC,f)$-maintainable}, for some complexity class $\calC$ and some function~\mbox{$f:\N\to\R$}, if there is a dynamic program $\prog$ and a $\calC$-algorithm $\mathbb A$ such that for each input database $\db$ over a domain of size $n$, each linear order $\leq$ on the domain, and each change sequence $\alpha$ of length $|\alpha| \leq f(n)$, the relation $Q$ in $\prog_\alpha(\state)$ and $\query(\alpha(\db))$ coincide, where $\state = (\inp, \mathbb A(\inp,\leq))$.

The following theorem is a slight adaption of Theorem 3 from \cite{Datta0SVZ17} and can be proved analogously. 

\begin{theorem}\label{theorem:fewerChanges}
    Every  $(\AC^i,\log^i n)$-maintainable, almost domain-independent query is in $\DynFOar$.
\end{theorem}

\subparagraph*{The Role of the Domain and Arithmetic} In order to focus on the study of changes of non-constant size, we choose a simplified approach and include arithmetic in our setting. We state our results for \DynFOar and according classes to make it clear that we assume the presence of a linear order, addition and multiplication relation on the whole domain at all times.\footnote{Different assumptions have been made in the literature. In \cite{PatnaikI94}, Patnaik and Immerman assume only a linear order to be present, while full arithmetic is assumed in \cite{PatnaikI97}. Etessami observed that arithmetic can be built up dynamically, and therefore subsequent work usually assumed initially empty auxiliary relations, see e.g. \cite{DattaKMSZ15, Datta0SVZ17}. In the setting of first-order incremental evaluation systems usually no arithmetic is assumed to be present \cite{DongST95}.}

We shortly discuss the consequences of not assuming built-in arithmetic on our results. For single tuple changes, the presence of built-in arithmetic essentially gives no advantage.

\begin{proposition}[{\cite[Theorem 4]{DattaKMSZ15}, formulation from \cite[Proposition 2]{Datta0SVZ17}}]\label{proposition:acisfo:almost}
  If a query $\query\in \DynFOar$ under single-tuple changes is almost domain-independent, then also $\query\in \DynFO$. 
\end{proposition}

This result relies on the fact that one can maintain a linear order and arithmetic on the activated domain in \DynFO under single-tuple changes \cite{Etessami98}, that is, on all elements that were in the active domain at some point of time. Under larger changes this is a priori not possible, as then one has to express in $\FO$ a linear order and arithmetic on the elements that enter the active domain.

An alternate approach to assuming the presence of built-in arithmetic is to demand that changes \emph{provide} additional information on the changed elements, for example, that they provide a linear order and arithmetic on the domain of the change. Using this approach, our results can be stated in terms of \DynFO and \DynFOMaj with the sole modification that sizes of changes are given relative to the size of the activated domain instead of with respect to the size of the whole domain. In this fashion our results also translate to the setting of first-order incremental evaluation systems of Dong, Su, and Topor \cite{DongST95}, where the domain can grow and shrink.

\section{Reachability under Multiple Changes}\label{section:reachability_woodbury}
In this section we prove that Reachability can be maintained under multiple changes. 

\reachabilitylognloglogn*

The approach is to use the well-known fact that Reachability can be reduced to the computation of the inverse of a matrix, and to invoke the Sherman-Morrison-Woodbury identity (cf.~\cite{HendersonS81}) to update the inverse. This identity essentially reduces the update of inverses after a change affecting $k$ nodes to the computation of an inverse of a $k \times k$ matrix. 

The challenge is to define the updates in $\FO(+, \times)$. The key ingredients here are to compute inverses with respect to many primes, and throw away primes for which the inverse does not exist. As, by Theorem \ref{theorem:fewerChanges}, it suffices to maintain the inverse for $\log^c n$ many steps for some $c$ to be fixed later (see proof of Theorem \ref{theorem:non_zero_inverse}), some primes remain valid if one starts from sufficiently  -- but polynomially -- many primes. We show that the inverse of $k \times k$ matrices over $\Z_p$ can be defined in $\FO(+, \times)$ for $k = \frac{\log n}{\log \log n}$.    

Theorem \ref{theorem:reachability_multiple_changes} in particular generalizes the result that Reachability can be maintained under single edge changes \cite{DattaKMSZ15}; our proof is an alternative to the proof presented in the latter work. In \cite{DattaKMSZ15}, maintenance of Reachability is reduced to the question whether a matrix $A$ has full rank, and it was shown that the rank %
can be maintained by storing and updating an invertible matrix $B$ and a matrix $D$ from which the rank can be easily extracted, such that $B \cdot A = D$.

\subsection{Reachability and Matrix Inverses}

There is a path from $s$ to $t$ in a graph $G$ of size $n$ with adjacency matrix $A_G$ if and only if the $s$-$t$-entry of the matrix $(nI - A_G)^{-1}$ is non-zero. This follows from the equation $(nI - A_G)^{-1} = \frac{1}{n} \sum_{i = 0}^{\infty} (\frac{1}{n} A_G)^i$ and the fact that $A_G^i$ counts the number of paths from $s$ to $t$ of length $i$.
Notice that $A \df nI - A_G$ is invertible as matrix over $\Q$ for every adjacency matrix $A_G$ since it is strictly diagonally dominant \cite[Theorem 6.1.10] {HornJ12}.

When applying a change $\Delta G$ to $G$ that affects $k$ nodes, the adjacency matrix of $G$ is updated by adding a suitable change matrix $\Delta A$ with at most $k$ non-zero rows and columns to $A$. Thus Theorem \ref{theorem:reachability_multiple_changes} follows from the following proposition\footnote{Due to lack of space some details are hidden here. The described reduction maps the empty graph to the matrix whose diagonal entries are $n$. Values of the inverse for this matrix cannot be determined in $\FO$, and thus one does not immediately get the desired result for Reachability. This issue can be circumvented by mapping to matrices with only some non-zero entries on the diagonal, and studying the inverse of the matrices induced by non-zero diagonal entries.}.

\begin{theorem} \label{theorem:non_zero_inverse}
  When $A \in \Z^{n \times n}$ takes values polynomial in $n$ and is assumed to stay invertible over $\Q$, then non-zeroness of entries of $A^{-1} \in \Q^{n \times n}$ can be maintained in $\DynFOar$ under changes that affect $\bigO(\frac{\log n}{\log \log n})$ rows and columns.
\end{theorem}

Each change affecting $\bigO(\frac{\log n}{\log \log n})$ rows and columns can be partitioned into constantly many changes that affect $k \df \frac{\log n}{\log \log n}$ rows and columns. We therefore concentrate on such changes in the following. 

The change matrix $\Delta A$ for a change affecting $k$ rows and columns has at most $k$ non-zero rows and columns and can therefore be decomposed into a product $UBV$ of suitable matrices $U, B,$ and $V$, where $U$, $B$, and $V$ have dimensions $n \times k$, $k \times k$, and $k \times n$, respectively.

\begin{lemma}\label{lemma:change_decomposition}
  Fix a ring $R$. Suppose $M \in R^{n \times n}$ with non-zero rows $r_{i_1}, \ldots, r_{i_k}$ and columns   $c_{j_1}, \ldots, c_{j_k}$. Then $M = UBV$ with $U \in R^{n \times k}, B \in R^{k \times k},$ and $V \in R^{k \times n}$ where 
  \begin{enumerate}
    \item $B$ is obtained from $M$ by removing all-zero rows and columns.
    \item $U = \begin{pmatrix}\tpl u_1 \\ \vdots \\ \tpl u_n \end{pmatrix}$ where $\tpl u_i = \begin{cases}
                                                                                                \tpl 0^T \text{ if } i \notin \{i_1, \ldots, i_k\}\\
                                                                                                \tpl e_m^T  \text{ if } i = i_m
                                                                                              \end{cases}$
    \item $V = \begin{pmatrix}\tpl v_1, \ldots, v_n \end{pmatrix}$ where $\tpl v_j = \begin{cases}
                                                                                                \tpl 0 \text{ if } j \notin \{j_1, \ldots, j_k\}\\
                                                                                                \tpl e_m  \text{ if } j = j_m
                                                                                              \end{cases}$
  \end{enumerate}
  Here, $\tpl e_m$ denotes the $m$-th unit vector. 
\end{lemma}

By the Sherman-Morrison-Woodbury identity (cf. \cite{HendersonS81}), the updated inverse can therefore be written as
\[(A + \Delta A)^{-1} = (A + UBV)^{-1} = A^{-1}-A^{-1}U(I+BVA^{-1}U)^{-1}BVA^{-1} \label{equation:woodbury} \tag*{$(\star)$}\] 
The inverse of a matrix in $\Z^{n \times n}$ with entries that are polynomial in $n$ is a matrix in $\Q^{n \times n}$ with entries $\frac{a}{b}$ that may involve numbers exponential in $n$. In particular computations cannot be performed in $\FO(+, \times)$ directly. For this reason all computations will be done modulo many primes, and non-zeroness of entries of $A^{-1}$ is extracted from these values. 

Let us first see how  to update $(A + \Delta A)^{-1}$ modulo a prime $p$ under the assumption that both $A \pmod p$ and $A + \Delta A \pmod p$ are invertible. Observe that $(I+BVA^{-1}U)^{-1}$ is a $k \times k$ matrix and therefore an essential prerequisite to compute $(A + \Delta A)^{-1} \pmod p$ is to be able to define the inverse of such small matrices. That this is possible follows from the following lemma and the fact that $[D^{-1}]_{ij} = (-1)^{i+j} \frac{\det D_{ji}}{\det D}$ for invertible $D \in \Z^{k \times k}_p$. Here $[C]_{ij}$ denotes the $ij$-th entry of a matrix $C$ and $C_{ji}$ denotes the submatrix obtained by removing the $j$-th row and the $i$-th column.

\begin{theorem}\label{theorem:determinant_in_fo}
  Fix a domain of size $n$ and a prime $p \in \bigO(n^c)$. The value of the determinant of a matrix $A \in \Z^{k \times k}_p$ for $k = \frac{\log n}{\log \log n}$ can be defined in $\FO(+, \times)$.
\end{theorem}

The technical proof of this theorem is deferred until the next Subsection \ref{section:determinant_in_fo}.

That $(A + \Delta A)^{-1} \pmod p$ can defined in $\FO(+, \times)$ using Equation \ref{equation:woodbury} now is a consequence of a straightforward analysis of the involved matrix operations.

\begin{proposition}\label{theorem:inverse_modp}
  Fix a domain of size $n$ and a prime $p \in \bigO(n^c)$. Given the inverse of a matrix $A \in \Z^{n \times n}_p$ and a matrix $\Delta A \in \Z^{n \times n}_p$ with at most $k = \frac{\log n}{\log \log n}$ non-zero rows and columns, one can determine whether $A + \Delta A$ is invertible in $\FO(+, \times)$ and, if so, the inverse can be defined.  
\end{proposition}
\begin{proof}
  A decomposition of the matrix $\Delta A$ into $UBV$ with $U \in \Z_p^{n \times k}, B \in \Z_p^{k \times k},$ and $V \in \Z_p^{k \times n}$ can be defined in $\FO(+, \times)$ using the characterization from Lemma \ref{lemma:change_decomposition}. A simple analysis of the right hand side of Equation \ref{equation:woodbury} -- taking the dimensions of $U, V,$ and $B$ into account -- yields that $VA^{-1}U$ and therefore $(I+BVA^{-1}U)^{-1}B$ are $k \times k$ matrices. Furthermore, $U(I+BVA^{-1}U)^{-1}BV$ is an $n \times n$ matrix that has at most $k$ non-zero rows and columns.
  
  The only obstacle to invertibility is that the inverse of $D \df I+BVA^{-1}U$ may not exist in $\Z_p$. This is the case if and only if $\det (D) \equiv 0 \pmod p$ which can be tested using Theorem~\ref{theorem:determinant_in_fo}. If $D$ is invertible, then its inverse can be defined by invoking  Theorem \ref{theorem:determinant_in_fo} twice and using $[D]_{ij} = (-1)^{i+j} \frac{\det D_{ji}}{\det D}$. 
  
  Finally, if one knows how to compute $(I+BVA^{-1}U)^{-1}$, each entry in $A^{-1}U(I+BVA^{-1}U)^{-1}BV$ can be defined by adding $k$ products of two numbers, and similarly for $(A^{-1}U(I+BVA^{-1}U)^{-1}BV)A^{-1}$. This can be done in $\FO(+, \times)$ due to Lemma \ref{lemma:computations_in_fo}.
\end{proof}

It remains to show how to maintain non-zeroness of entries of $(A + \Delta A)^{-1} \in \Q^{n \times n}$. Essentially a dynamic program can maintain a Chinese remainder representation of $(A + \Delta A)^{-1}$ and extract whether an entry is non-zero from this representation. An obstacle is that whenever $(I+BVA^{-1}U)^{-1} \pmod p$ does not exist for a prime $p$ during the update process, then this prime $p$ becomes \emph{invalid} for the rest of the computation. The idea to circumvent this is simple: with each change, only a small number of primes become invalid. However, since the determinant can be computed in $\NC^2$ (cf. \cite{Cook85}), using Theorem \ref{theorem:fewerChanges} we only need to be able to maintain a correct result for $\log^2 n$ many steps. Thus starting from sufficiently many primes will guarantee that enough primes are still valid after $\log^2 n$ steps.

We make these numbers more precise in the following.

\begin{proofof}{Theorem \ref{theorem:non_zero_inverse}}
  By  Theorem \ref{theorem:fewerChanges} and since values of the inverse of a matrix are almost domain-independent,  it suffices to exhibit a dynamic program\footnote{Actually we only describe a program that works correctly for sufficiently large $n$. However, small $n$ can be easily dealt with separately.} that maintains non-zeroness of entries of $A^{-1}$ for $\log^2 n$ changes of size $\frac{\log n}{\log \log n}$.  The dynamic program maintains $A^{-1} \pmod p$ for each of the first $2n^3$ many primes $p$, which, by the Prime Number Theorem, can be found among the first $n^4$ numbers. Denote by $P$ the set of the first $2n^3$ primes. The $\NC^2$ initialization procedure computes $A^{-1} \pmod p$ for each prime in $P$. The update procedure for a change $\Delta A$ is simple:
  \begin{enumerate}
    \item[(1)] For each prime $p \in P$:
      \begin{enumerate}
        \item[(a)] If $(A + \Delta A)^{-1} \pmod p$ is not invertible then remove $p$ from $P$.
        \item[(b)] If $(A + \Delta A)^{-1} \pmod p$ is invertible then update $(A + \Delta A)^{-1} \pmod p$.
      \end{enumerate}
    \item[(2)] Declare $[(A + \Delta A)^{-1}]_{st} \neq 0$  if there is a prime $p \in P$ with $[(A + \Delta A)^{-1}]_{st} \not\equiv 0 \pmod p$.
  \end{enumerate}
  The Steps 1a and 1b can be performed in $\FO(+, \times)$ due to Proposition \ref{theorem:inverse_modp}. 
  
It remains to argue that the result from Step 2 is correct. Observe that the values of entries of $A$ are at most $n$ at all times, and therefore $\det(A) \leq n! n^n \leq 2^{n^2}$ for large enough~$n$.  Thus, since $\det(A) \neq 0$ over $\Z$ by assumption, there are at most $n^2$ primes $p$ such that $\det(A) \equiv  0 \pmod p$, for all $A$ reached after a sequence of changes. 
  
  In particular, $(A + \Delta A)^{-1} \pmod p$ is not invertible --- equivalently, $(I+BVA^{-1}U)^{-1} \pmod p$ does not exist --- for at most $n^2$ primes $p$. 
Hence, each time Step 1 is executed, at most $n^2$ primes are declared invalid and removed from $P$. All in all this step is executed at most $\log^2 n$ times, and therefore not more than $n^3$ primes are removed from~$P$. 
Thus for the remaining $n^3$ valid primes, the inverses $(A + \Delta A)^{-1} \pmod p$ are computed correctly.

Each entry of $(A + \Delta A)^{-1}$ is, again, bounded by $2^{n^2}$, so if $[(A + \Delta A)^{-1}]_{st} \neq 0$ there are at most $n^2$ primes $p \in P$ with $[(A + \Delta A)^{-1}]_{st} \equiv 0 \pmod p$. So, the result declared in Step~2 is correct.  
\end{proofof}

\subsection{Defining the Determinant of Small Matrices}\label{section:determinant_in_fo}

In this subsection we prove Theorem \ref{theorem:determinant_in_fo}. The symbolic determinant of a $k \in \bigO(\frac{\log}{\log \log n})$ sized matrix is a sum of $k! \in n^{\bigO(1)}$ monomials and therefore cannot be na\"{\i}vely defined in $\FO(+, \times)$. Here we use the fact that $\FO(+, \times)$ can easily convert $\log n$ bit numbers into their Chinese remainder presentation and back, and show how the determinant can be computed modulo $\log \log n$ bit primes.

It is easy to verify whether the value of a determinant modulo a $\bigO(\log \log n)$ bit prime is zero
in $\FO(+, \times)$ by guessing a linear combination witnessing that the rank is less than full. We aim for a characterization that allows to reduce the verification of determinant values to such zeroness tests. To this end we use the self-reducibility and multilinearity of determinants.  Assume $[A]_{11} \neq 0$ and that the determinant of $A_{11}$ is also non-zero. Then the determinant can be written as $[A]_{11}\cdot d + r$ for some $d$ and $r$. By finding an $a$ such that the determinant is zero when $[A]_{11}$ is replaced by $a$ in $A$ we gain $r = -ad$. Repeating this step recursively for $d$ ---  which is the determinant of a smaller matrix --- one obtains a procedure for determining the value of the determinant that can be parallelized.

The following lemma is a preparation for deriving the characterization. We denote by $A_i$ the matrix obtained from a matrix $A$ by removing all rows and columns larger than $i$.

\begin{lemma}\label{lemma:permutations}
  Suppose $B = (\tpl b_1, \ldots, \tpl b_k) \in \F^{k \times k}$ is a non-singular matrix over a field $\F$. Then there is a permutation $\pi: [k] \rightarrow [k]$ such that for $A \df (b_{\pi(1)}, \ldots, b_{\pi(k)})$:
  \begin{quote}
    $[A]_{ii} \neq 0$ and  $\det(A_{i}) \neq 0$ for all $i \in [k]$
  \end{quote}
\end{lemma}
\onlyLong{
\begin{proof}
  In the Laplacian expansion $\sum_{j=1}^k (-1)^{k+j} [B]_{kj}\det (B_{kj})$ of $\det(B)$ with respect to the $k$-th row there must be at least on non-zero term; say, the $\ell$-th term. Then $[B]_{k\ell} \neq 0$ and $\det(B_{k\ell}) \neq 0$. Thus if $B'$ is the matrix obtained by swapping the $k$-th and $\ell$-th columns of $B$ then $[B']_{kk} \neq 0$ and, if $k > 1$, $\det(B'_{k-1}) \neq 0$. Proceed inductively with the matrix $B'$, and combine the column swaps into a permutation $\pi$.
\end{proof}
}

The following proposition characterizes the determinant of a matrix. We will see that this characterization allows for parallel computation of the determinant of small matrices.

\begin{proposition}\label{proposition:determinant_characterization}
  Suppose $A = (a_{ij})_{1 \leq i, j \leq k} \in \F^{k \times k}$ is a matrix over a field $\F$ such that $a_{ii} \neq 0$ and $\det(A_{i}) \neq 0$ for all $i \in [k]$. Let $A^b_{i}$ be the matrix obtained from $A_{i}$ by replacing $a_{ii}$ by $b$ for some $b \in \F$. Then there are unique $b_2, \ldots, b_k \in \F$ and $d_1, \ldots, d_k \in \F$ such that
  \begin{enumerate}
    \item $d_1 = a_{11}$,
    \item $d_{i} = (a_{ii}- b_i) d_{i-1}$, and 
    \item $\det(A^b_i) = 0$ 
  \end{enumerate}
  for $2 \leq i \leq k$. Furthermore, it holds that $d_i = \det(A_i)$.
\end{proposition}
\onlyLong{
\begin{proof}
Clearly, $d_1 = \det(A_1)$. We inductively show that the $b_i$ exist and are unique. The values $d_i$ are then determined by (b), and we prove that $d_i = \det(A_i)$ for $i \in [2, \ldots, k]$. 
  Suppose this has been ensured for $i-1$. Expanding the determinant of $A_{i}$ with respect to the $i$-th row and splitting the sum into the term for the $i$-th column and the term for all other columns yields
    \[\det(A_i) = a_{ii} d_{i-1} + r_i \label{equation:e1} \tag*{$(\star)$}\]
  with $r_i \df \sum_{j=1}^{i-1} (-1)^{i+j} a_{ij} \det ((A_i)_{ij})$.

  Similarly the determinant $d_{i}(b) \df \det(A^b_i)$ expands to $d_{i}(b)=  b d_{i-1} + r_i$. Since $d_{i-1} = \det(A_{i-1}) \neq 0$ there is a unique $b_i$ such that $d_{i}(b_i) = 0$. With this $b_i$, we have that $r_i = -b_i d_{i-1}$, and plugging this into Equation \ref{equation:e1} yields that  $d_i = \det(A_i)$.
\end{proof}
}

Finally we show that the characterization from the previous proposition can be used to define the determinant of small matrices in $\FO(+, \times)$.

\begin{proofof}{Theorem \ref{theorem:determinant_in_fo}}
  Suppose $A \in \Z^{k \times k}_p$ is a matrix with $p \in \bigO(n^c)$ and $k = \frac{\log n}{\log \log n}$. The idea is to define $\det(A) \pmod p$ in Chinese remainder representation for primes $q_1, \ldots, q_m$. A simple calculation shows that $m \in \bigO(\log n)$ primes each of $\bigO(\log \log n)$ bits suffice. The Chinese remainder representation can be defined from $A$ and the value $\det(A) \pmod p$ can be recovered from the values $\det(A) \pmod{q_1}, \ldots, \det(A) \pmod{q_m}$ in $\FO(+, \times)$ due to Lemma~\ref{lemma:computations_in_fo}.
  Thus let us show how to define $\det(A) \pmod q$ for a prime $q$ of $\bigO(\log \log n)$ bits. 
  
  The idea is to first test whether the determinant is zero. If not, the fact that it is not zero is used to define the determinant using Proposition~\ref{proposition:determinant_characterization}. 
  
  If $A \pmod q$ is singular then there exists a non-trivial linear combination of the columns that yields the all zero vector. Such a linear combination is determined by specifying one $\bigO(\log \log n)$ bit number for each of the $k$ columns. It can thus be encoded in $\bigO(\log n)$ bits, and therefore existentially quantified by a first-order formula. Such a ``guess'' can be decoded (i.e., the $k$ numbers of $\bigO(\log \log n)$ length can be extracted) in $\FO(+, \times)$, see Section \ref{section:preliminaries}. Checking if a guessed linear combination is zero requires to sum $k$ small numbers and is  hence in $\FO(+, \times)$ due to Lemma~\ref{lemma:computations_in_fo}.

  Now, for defining the determinant $\det (A)\pmod q$ when $A \pmod  q$ is non-singular, a formula can guess a permutation $\pi$ of $[k]$ and verify that it satisfies the conditions from Lemma \ref{lemma:permutations}. Note that such a permutation can be represented as a sequence of $k$ pairs of numbers of $\log \log n$ bits each, and hence be stored in $\bigO(\log n)$ bits. The verification of the conditions from Lemma \ref{lemma:permutations} requires the zero-test for determinants explained above. After fixing $\pi$, the values $b_2, \dots, b_k$ as well as $d_1, \ldots, d_{k}$ from Proposition~\ref{proposition:determinant_characterization} can be guessed and verified. Again, these numbers can be stored in $\bigO(\log n)$ bits. For verifying the conditions from Proposition~\ref{proposition:determinant_characterization} on the determinants of $A^b_i$, the zero-test for determinants is used. 
\end{proofof}

\section{Distances under Multiple Changes}\label{section:distances_woodbury}

In this section we extend the techniques from the previous section to show how distances can be maintained under changes that affect polylogarithmically many nodes with first-order updates that may use majority quantifiers. Afterwards we discuss how the techniques extend to other dynamic complexity classes. 

\distancespolylogn*

The idea is to use generating functions for counting the number of paths of each length, following Hesse \cite{Hesse03}. Fix a graph $G$ with adjacency matrix $A_G \in \Z^{n \times n}$ and a formal variable~$x$. Then $D \df \sum_{i = 0}^{\infty} (xA_G)^i$ is a matrix of formal power series from $\Z[[x]]$ such that if $[D]_{st} = \sum_{i = 0}^{\infty} c_i x^i$ then $c_i$ is the number of paths from $s$ to $t$ of length $i$. In particular, the distance between $s$ and $t$ is the smallest $i$ such that $c_i$ is non-zero.
Note that if such an $i$ exists, then $i < n$.

Similarly to the corresponding matrix from the previous section, the matrix $D$ is invertible over $\Z[[x]]$ and can be written as $(I - x A_G)^{-1}$ (cf.~\cite[Example 3.6.1]{Godsil93}). The maintenance of distances thus reduces to maintaining for a matrix $A \in \Z[[x]]$, for each entry $(s, t)$, the smallest $i < n$ such that the $i$th coefficient is non-zero.

\begin{theorem} \label{theorem:inverse_coefficients}
  Suppose $A \in \Z[[x]]^{n \times n}$ stays invertible over $\Z[[x]]$. For all $s, t \in [n]$ one can maintain the smallest $i < n$ such that the $i$th coefficient of the $st$-entry of $A^{-1}$ is non-zero  in $\DynFOMaj(+,\times)$ under changes that affect  $\bigO(\log^c n)$ nodes, for fixed $c \in \N$. 
\end{theorem}

The idea is the same as for Reachability. When updating $A$ to $A + \Delta A$ then one can decompose the change matrix $\Delta A$ into $UBV$ for suitable matrices $U, B,$ and $V$, and apply the Sherman-Morrison-Woodbury identity \ref{equation:woodbury}, this time over the field of fractions $\Z((x))$\onlyLong{ (see the appendix for a short recollection of this field)}.

Of course computing with inherently infinite formal power series is not possible in $\DynFOMaj(+,\times)$. However, as stated in Theorem \ref{theorem:inverse_coefficients}, in the end we are only interested in the first $i < n$ coefficients of power series. We therefore show that it suffices to truncate all occurring power series at the $n$-th term and use $\FOMaj(+, \times)$'s ability to define iterated sums and products of polynomials \cite{HesseAB02}.  

Formally, we have to show that no precision for the first $i < n$ coefficients is lost when computing with truncated power series. This motivates the following definition. A formal power series $g(x) = \sum_i c_i x^i \in \Z[[x]]$ is an \emph{$m$-approximation} of a formal power series $h(x) = \sum_i d_i x^i \in \Z[[x]]$, denoted by $g(x) \approx_m h(x)$,  if $c_i = d_i$ for all $i \leq m$. This notion naturally extends to matrices over $\Z[[x]]$: a matrix $A \in \Z[[x]]^{\ell \times k}$ is an $m$-approximation of a matrix $B \in \Z[[x]]^{\ell \times k}$ if each entry of $A$ is an $m$-approximation of the corresponding entry of~$B$. The notion of $m$-approximation is preserved under all arithmetic operations that will be relevant.
 
\begin{lemma}\label{lemma:approximation_closure} Fix an $m \in \N$.
  \begin{enumerate}
    \item  Suppose $g(x), g'(x), h(x), h'(x)  \in \Z[[x]]$ with $g(x) \approx_m g'(x)$ and $h(x) \approx_m h'(x)$. Then
      \begin{enumerate}[(i)]
          \item $g(x) + h(x) \approx_m g'(x) + h'(x)$,
          \item $g(x)h(x) \approx_m g'(x)h'(x)$, and
          \item $\frac{1}{g(x)} \approx_m \frac{1}{g'(x)}$ whenever $g(x)$ and $g'(x)$ are normalized.
      \end{enumerate}
    \item Suppose $A, A', B, B' \in \Z[[x]]^{n \times n}$ with $A \approx_m A'$ and $B \approx_m B'$. Then
      \begin{enumerate}[(i)]
        \item $A + B \approx_m A' + B'$,
        \item $AB \approx_m A'B'$,
        \item If $A$ is invertible over $\Z[[x]]$ then so is $A'$, and $A^{-1} \approx_m A'^{-1}$.
      \end{enumerate}
  \end{enumerate}
\end{lemma}
Here, a formal power series $\sum_i c_i x^i \in \Z[[x]]$ is \emph{normalized} if $c_0 = 1$.

An approximation of the inverse of a matrix $A \in \Z[[x]]^{n \times n}$ can be updated using the Sherman-Morrison-Woodbury identity.
\begin{proposition}\label{proposition:approximation_update}
    Suppose $A \in \Z[[x]]^{n \times n}$ is invertible over $\Z[[x]]$, and $C \in \Z[[x]]^{n \times n}$ is an $m$-approximation of $A^{-1}$. If $A + \Delta A$ is invertible over $\Z[[x]]$ and $\Delta A$ can be written as $UBV$ with $U \in \Z[[x]]^{n \times k}, B \in \Z[[x]]^{k \times k},$ and $V \in \Z[[x]]^{k \times n}$, then
      $$(A + \Delta A)^{-1} \approx_m C-CU(I+BVCU)^{-1}BVC$$
\end{proposition}
\begin{proof}
   This follows immediately from the Sherman-Morrison-Woodbury identity $(A + UBV)^{-1} = A^{-1}-A^{-1}U(I+BVA^{-1}U)^{-1}BVA^{-1}$ and Lemma \ref{lemma:approximation_closure}. 
\end{proof}

As already discussed in Section \ref{section:reachability_woodbury}, the Sherman-Morrison-Woodbury identity involves inverting $k \times k$ matrices, which reduces to computing the determinant of such matrices.  We show that this is possible in $\FOMaj$ for $k \times k$ matrices of polynomials for $k \in \bigO(\log^c n)$. 

\begin{lemma}\label{lemma:determinant_polylog}
  Fix a domain of size $n$ and $c \in \N$. The  determinant of a matrix $A \in \Z[x]^{k \times k}$, with entries of degree polynomial in $n$, can be defined in $\FOMaj(+, \times)$ for $k \in \bigO(\log^c n)$.
\end{lemma}
\begin{proof}
  We show that the value can be computed in uniform $\TC^0$, which is as powerful as $\FOMaj(+, \times)$ \cite{BarringtonIS90}.

  Computing the determinant of an $k \times k$ matrix is equivalent to computing the iterated matrix product of $k$ matrices of dimension at most $(k+1) \times (k+1)$ \cite{Cook85}, and this reduction is a uniform $\TC^0$-reduction as can be seen implicitly in \cite[p.~482]{MahajanV99}. Thus the lemma statement follows from the fact that iterated products of matrices $A_1, \ldots, A_k \in \Z[x]^{k \times k}$ with $k \in \bigO(\log^c n)$ can be computed in uniform $\TC^0$, which can be proven like in~\mbox{\cite[p.~69]{AgrawalV08}}. \onlyLong{The full proof can be found in the appendix.}
\end{proof}

\begin{proofof}{Theorem \ref{theorem:inverse_coefficients}}
  The dynamic program maintains an $n$-approximation $C \in \Z[x]^{n \times n}$ of $A^{-1}$ that truncates $A^{-1}$ at degree $n$. When $A$ is updated to $A + \Delta A$ then:
  \begin{enumerate}
   \item $\Delta A$ is decomposed into suitable $U \in \Z[x]^{n \times k}, B \in \Z[x]^{k \times k},$ and $V \in \Z[x]^{k \times n}$;
   \item $C$ is updated via $C' \df C-CU(I+BVCU)^{-1}BVC$;
   \item All entries of $C'$ are truncated at degree $n$.
  \end{enumerate}
  
  The steps can be defined in $\FOMaj(+, \times)$ due to Lemma \ref{lemma:change_decomposition}, Lemma \ref{lemma:determinant_polylog}, and the fact that iterated addition and multiplication of polynomials can be defined in $\FOMaj(+, \times)$, see~\cite{HesseAB02}. The maintained matrix $C$ is indeed an $n$-approximation of $A^{-1}$ due to Proposition~\ref{proposition:approximation_update}.
\end{proofof}

From the proof of Theorem \ref{theorem:inverse_coefficients} it is clear that the main obstacle towards maintaining distances for changes that affect a larger set of nodes is to compute determinants of larger matrices. Since distances can be computed in $\NL$, only classes below $\NL$ are interesting from a dynamic perspective. As an example we state a result for the circuit class $\NC^1$.

\begin{corollary}\label{corollary:distances_NC1}
  Reachability and Distance can be maintained in $\DynNC^1$ under changes that affect $\bigO(2^{\sqrt{\log n/\log^*n}})$ nodes.   
\end{corollary}

Here $\log^*n$ denotes the smallest number $i$ such that $i$-fold application of $\log$ yields a number smaller than $1$. \onlyLong{The corollary follows by plugging Lemma \ref{lemma:distance_NC1} into the proof above. %

\begin{lemma}\label{lemma:distance_NC1}
Fix a domain of size $n$. The determinant of a matrix $B \in \Z[x]^{k \times k}$, with entries of degree polynomial in $n$, can be computed in uniform $\NC^1$ for $k \in \bigO(2^{\sqrt{\log n/\log^*n}})$. 
\end{lemma} 
}

\section{Conclusion}\label{section:conclusion}
For us it came as a surprise that Reachability can be maintained under changes of non-constant size, without any structural restrictions.
In contrast, the dynamic program for Reachability from \cite{DattaKMSZ15} can only deal with changing $\log n$ many outgoing edges of single nodes (or, symmetrically, $\log n$ many incoming edges; a combination is not possible). \onlyLong{For that program it is essential that only single rows of the adjacency matrix are changed. }

It would be interesting to improve our results for $\DynFO(+, \times)$ to changes of size~$\bigO(\log n)$. The obstacle is the computation of determinants of matrices of this size, which we can only do for $\bigO(\frac{\log n}{\log \log n})$ size matrices. Yet in principle our approach can deal with certain changes that affect more nodes: the matrices $U$ and $V$ in the Sherman-Morrison-Woodbury identity can be chosen differently, as long as all computations involve only adding $\bigO(\log n)$ numbers.  

One of the big remaining open questions in dynamic complexity is whether distances are in $\DynFO$. Our approach sheds some light\onlyLong{ on this question}. It can be adapted so as to maintain information within $\DynFO(+, \times)$ from which shortest distances can be extracted in $\FOMaj(+, \times)$. \onlyLong{The technical proof of this result is deferred to the appendix.}

\begin{theorem}\label{theorem:towards_distances_in_dynfo}
Distances can be defined by a $\FOMaj(+, \times)$ query from auxiliary relations that can be maintained in $\DynFO(+, \times)$ under changes that affect $\bigO(\frac{\log n}{\log \log n})$ nodes.
\end{theorem}

\bibliography{bibliography}

\onlyLong{
  \section*{Appendix}
\section{Background on Formal Power Series}
Recall that $\Z$ is an integral domain and has the only units $1$ and $-1$. By $\Z[[x]]$ we denote the ring of formal power series over $\Z$, i.e. the ring with elements $\sum_i c_i x^i$ and natural addition an multiplication. An element $\sum_i c_i x^i \in \Z[[x]]$ is \emph{normalized} if $c_0 = 1$. Since $\Z$ is an integral domain, all normalized elements of $\Z[[x]]$ have an inverse. The integral domain $\Z[[X]$ can be embedded into its field of fractions $\Z((x)) \df \{\frac{g(x)}{h(x)} \mid g(x), h(x) \in Z[[x]] \text{ and } h(x) \neq 0$. We denote the subring of $\Z[[x]]$ consisting of all finite polynomials by $\Z[x]$; and the field of fractions of $\Z[x]$ by $\Z(x)$.

A matrix $A \in \Z[[x]]^{n \times n}$ is invertible over $\Z[[x]]$ if there is a matrix $B \in  \Z[[x]]^{n \times n}$ with $AB = I$. The matrix $A \in \Z[[x]]^{n \times n}$ is invertible over $\Z[[x]]$ if and only if it is invertible in $\Z((x))$ and the constant term of $\det(A)$ is a unit of $\Z$, i.e. it is $1$ or $-1$. 

For a polynomial $g(x) \in \Z[x]$ we abbreviate its degree by $\deg g(x)$ and write $\|g(x)\|$ for the value of its largest coefficient. The degree of a representation $\frac{g(x)}{h(x}$ of an element of $\Z(x)$ is the maximum of the degrees of $g(x)$ and $h(x)$, and similarly for the largest coefficient. Degree and maximal coefficient are defined similarly for matrices over $\Z[x]$ and $\Z(x)$

If $A$ is of the form $I + xC$ for some matrix $C \in \Z[[x]]^{n \times n}$ then $A$ is invertible over $\Z[[x]]$ as $\det(I + xC)$ exists and is a normalized polynomial.

\section{Proofs of Section \ref{section:distances_woodbury}}

\begin{proofof}{Lemma \ref{lemma:approximation_closure}}
  The first two parts of (a) are straightforward. For the last part suppose that $\frac{1}{g(x)} = \sum_i d_i x^i$ and $\frac{1}{g'(x)} = \sum_i d'_i x^i$. Further write $g(x)$ and $g'(x)$ as $g(x) = \sum_{i=0}^m c_i x^i + r(x)$ and $g'(x) = \sum_{i=0}^m c_i x^i + r'(x)$ where $x^{m+1} | r(x)$ and $x^{m+1} | r'(x)$. Then it is easy to see that $c_0 d_0 = 1$ and $\sum_{i=0}^j c_i d_{j-i} = 0$ for all $j \in [m]$. Similarly $c_0 d'_0 = 1$ and $\sum_{i=0}^j c_i d'_{j-i} = 0$. Solving both systems of equations yields $d_i = d'_i$ for $i \in [m]_0$.
  
  The first two parts of (b) follow immediately from (a). For the third part, recall that $A$ is invertible over $\Z[[x]]$ if and only if $\det(A) \neq 0$ and is normalized. This translates, via (a), to the matrix $A'$. Furthermore \[[A^{-1}]_{st} = (-1)^{s+t}\frac{\det(A_{ts})}{\det(A)} \approx_m (-1)^{s+t}\frac{\det(A'_{ts})}{\det(A')} = [A'^{-1}]_{st}\]
\end{proofof}

\begin{proofof}{Lemma \ref{lemma:determinant_polylog}}
  We show that the value can be computed in uniform $\TC^0$, which is as powerful as $\FOMaj(+, \times)$ \cite{BarringtonIS90}.

  Computing the determinant of an $k \times k$ matrix is equivalent to computing iterated matrix product of $k$ matrices of dimension at most $(k+1) \times (k+1)$ \cite{Cook85}, and this reduction is indeed a uniform $\TC^0$-reduction as can be seen implicitly in cf.~\cite[p.~482]{MahajanV99}. Thus the lemma statement follows from the fact that iterated products of matrices $A_1, \ldots, A_k \in \Z[x]^{k \times k}$ with $k \in \bigO(\log^c n)$ can be computed in uniform $\TC^0$, which can be proven in the spirit of \cite[p.~69]{AgrawalV08}. 
  
  For the sake of completeness we outline the proof. We first explain how $\bigO(\sqrt {\log n})$ such matrices can be multiplied.  Each entry in such a product is the sum of $(\log^c n)^{\bigO(\sqrt {\log n})}$ many products of $\bigO(\sqrt {\log n})$ polynomials. Such products can be computed in uniform $\TC^0$ due to \cite[Corollary 6.5]{HesseAB02}. The sum can be computed in $\TC^0$ as it is over at most polynomially many terms:
 \begin{align*}
 (\log^c n)^{\bigO(\sqrt {\log n})} = 2^{\log((\log^c n)^{\bigO(\sqrt {\log n})})} = 2^{\bigO(\sqrt {\log n}) \log \log n} \subseteq 2^{\bigO(\sqrt {\log n}) \bigO(\sqrt {\log n})} = n^{\bigO(1)}.
\end{align*}

  The idea for computing the product of $A_1, \ldots, A_k$ with $k \in \bigO(\log^c n)$ is to partition the sequence into fragments of length $\bigO(\sqrt {\log n})$ each. The product $B_i$ of the matrices of the $i$th fragment can be computed by the procedure from above. As $k \in \bigO(\log^c n)$ there are $\bigO(\log^{c-\frac{1}{2}} n)$ such fragments. This procedure can now be repeated recursively, that is, the sequence $B_i$ is partitioned into fragments of length $\bigO(\sqrt {\log n})$, and so on. After $2c$ repetitions, the final product is obtained.

\end{proofof}

\begin{proofsketchof}{Lemma \ref{lemma:distance_NC1}}
  In order to keep the depth of the $\NC^1$ circuit in $\bigO(\log n)$, the circuit reduces the amount of multiplications of polynomials by computing $\det(B(a_1)), \ldots, \det(B(a_n))$ for distinct integers $a_1, a_2, \ldots, a_{n^2}$. Here $B(a_i)$ denotes the matrix $B$ evaluated at $a_i$. As the determinant of $B$ has degree at most $n^2$, it can be recovered from these values in the end by using interpolation. Interpolating a polynomial is even possible in uniform $\TC^0$, see \cite{HesseAB02}. 

  We follow the outline of Lemma \ref{lemma:determinant_polylog} and use that it suffices to show that iterated products of matrices $A_1, \ldots, A_k \in \Z^{ k\times k}$ with $k \in \bigO(\log^c n)$ can be computed in uniform $\NC^1$.

  Building a tree to compute the product of the matrices two at a time gives a semi-unbounded fan-in arithmetic formula of depth $\log k$ with $+$-fan-in $2r^2$ and $\times$-fan-in $2$. The $+$-fan-in can be reduced to $2$ only to make the circuit depth $(\log r)^2$ by introducing a binary arithmetic formula of depth $O(\log r)$ at each unbounded fan-in $+$-gate.

  We know that $O(\log n)$ depth arithmetic circuits, and therefore equivalently, $O(\log n)$ depth arithmetic formulas can be computed by polynomial size uniform $\NC^1$-circuits of depth $O(\log n \log^*n)$ using Jung's theorem \cite{Jung85} (see Allender's survey \cite{Allender04} for a simple proof). Hence, arithmetic formulas of depth $O(\log n/\log^*n)$ can be computed in $\NC^1$. Thus if we have $(\log k)^2 = O(\log n/\log^*n)$ then $\det (B(a_i))$ is in $\NC^1$, which yields $k = O(2^{\sqrt{\log n/\log^*n}})$. 

\end{proofsketchof}

\section{Proofs of Section \ref{section:conclusion}}

Towards proving Theorem \ref{theorem:towards_distances_in_dynfo} we proceed in the same spirit as for maintaining distances under polylogarithmic changes in $\DynFOMaj(\times, +)$, see Section \ref{section:distances_woodbury}, and prove the following. 

\begin{theorem} \label{theorem:inverse_coefficients_tcquery}
  Suppose $A \in \Z[[x]]^{n \times n}$ contains coefficients that are polynomial in $n$ and stays invertible over $\Z[[x]]$. For all $s, t \in [n]$ one can maintain auxiliary relations in $\DynFO(+, \times)$ under changes that affect  $\bigO(\frac{\log n}{\log \log n})$ nodes, from which the smallest $i < n$ such that the $i$-th coefficient of the $st$-entry of $A^{-1}$ is non-zero can be defined in $\FOMaj(+, \times)$.  
\end{theorem}

The approach is the same as before. However, truncating the approximated polynomials does not suffice here as in $\FO(+, \times)$ it is not possible to compute with polynomials of large degree and large coefficients.

Therefore our goal is to maintain an implicit representation of an $n$-approximation $C(x)$ of $A^{-1} \in \Z[[x]]$ from which the smallest non-zero term of each of the entries can be extracted. The idea is to store and update the evaluation $C(a)$ for several numbers $a \in \N$.  If the entries of $C(x)$ have small degree, then the smallest non-zero term can be extracted from this via Cauchy interpolation. However, when only storing $C(x)$ implicitly via $C(a)$ it is not possible to truncate the polynomials after each step, as in the proof of Theorem \ref{theorem:inverse_coefficients}. Furthermore, the update formula for $C(x)$ provided in Lemma \ref{proposition:approximation_update} does not ensure that polynomials keep a small degree if they are not truncated. 

For this reason we proceed as follows. We first introduce a representation where each entry of $C(x)$ is represented by a fraction $\frac{g(x)}{h(x)}$ such that $g(x)$ and $h(x)$ are polynomials of degree $\bigO(n^d)$ for some $d$. The program then stores $g(a)$ and $h(a)$ for several numbers $a \in \N$. Actually the numbers $g(a)$ and $h(a)$ might be very large, indeed exponential in $n$, and therefore we will store all numbers in Chinese remainder representation.

Next we prepare by proving several lemmata that will ensure the correctness of this course of action. Afterwards we prove Theorem \ref{theorem:inverse_coefficients_tcquery} by presenting the dynamic program in detail. 

We start by introducing a representation of the matrix $C(x)$ in terms fractions. A quotient $m$-approximation of a formal power series $f(x)$ is a fraction $\frac{g(x)}{h(x)} \in \Z(x)$ with normalized $h(x)$ such that $\frac{g(x)}{h(x)} \approx_m f(x)$ when $\frac{g(x)}{h(x)}$ is treated as a formal power series. Quotient $m$-approximations for matrices are defined analogously.

\begin{lemma}\label{lemma:quotient_approximation_update}
    Suppose $A \in \Z[[x]]^{n \times n}$ is invertible over $\Z[[x]]$, and $C \in \Z(x)^{n \times n}$ is a quotient $m$-approximation of $A^{-1}$. Then if $A + \Delta A$ is invertible over $\Z[[x]]$ and $\Delta A$ can be written as $UBV$ with $U \in \Z[x]^{n \times k}, B \in \Z[x]^{k \times k},$ and $V \in \Z[x]^{k \times n}$ then \[C' \df C-CU(I+BVCU)^{-1}BVC\] is an $m$-approximation of $(A + \Delta A)^{-1}$.
    Furthermore, if $C'$ is computed with Algorithm \ref{algorithm:woodbury}, $B$ is of the form $xB^*$ for some $B^* \in \Z[x]^{k \times k}$ that takes only values from $\{-1, 0, 1\}$, and $U,V$ take only values from $\{0,1\}$, then
    \begin{enumerate}
      \item $\deg C' \in \bigO(k^3 \deg C)$ and $\|C'\| \in (\|C\| k \deg C)^{\bigO(k^3)}$, and
      \item if the denominators in $C$ are normalized then the denominators of $C'$ are normalized as well. 
    \end{enumerate}
\end{lemma}
\begin{proof}
  That $C'$ is an $m$-approximation of $(A + \Delta A)^{-1}$ follows from Proposition \ref{proposition:approximation_update} and the assumption that $C$ is a quotient $m$-approximation. If all denominators of $C$ are normalized and $B$ is of the form $xB^*$ then all computations preserve that intermediate denominators are normalized.  
  
  It remains to show the bounds on $\deg C'$ and $\|C'\|$. Note that for polynomials $h_1, \ldots, h_k \in \Z[x]$ one has $\deg(\sum_i h_i) = \max_i{\deg(h_i)}$ and $\deg(\prod_i h_i) \in  \bigO(\sum_i \deg(h_i))$ when $\sum_i h_i$ and $\prod_i h_i$ are computed in the naive way. Hence in Algorithm \ref{algorithm:woodbury},  $E$ has degree $\bigO(\deg C)$ and $f(x)$ has degree $\bigO(k^2 \deg C)$, and therefore $\det(f(x)E)$ and $(I+BVCU)^{-1}$ have degree $\bigO(k^3 \deg C)$. It follows that $\deg C' \in \bigO(k^3 \deg C)$. The estimation of $\|C'\|$ is similar, using the facts that $\|\sum_i h_i\| \in \bigO(\sum_i \|h_i\|)$ and $\| \prod_i h_i\| \in \bigO(\prod_i{\deg(h_i)} \prod_i{\|h_i\|})$. 
\end{proof}

\begin{algorithm}
\caption{Updating a quotient $m$-approximation}\label{algorithm:woodbury}
\begin{algorithmic}[1]
\Input Matrices $C \in \Z(x)^{n \times n}, U \in \Z[x]^{n \times k}, B \in \Z[x]^{k \times k},$ and $V \in \Z[x]^{k \times n}$ 
\Output Matrix $C' \df C-CU(I+BVCU)^{-1}BVC$
\State Compute $E = I+BVCU\in \Z(x)^{k \times k}$
\State Compute the product $f(x)$ of all denominators of $E$. 
\For {all $i, j \leq k$}
  \State Compute the $ij$th entry of the inverse $E^*$ of $f(x)E$ as $(-1)^{i+j}\frac{\det(f(x)E_{ji})}{ \det(f(x)E)}$.
\EndFor
\State Compute the inverse of $E$ as $f(x)E^*$.
\State Compute $C'$.  
\end{algorithmic}
\end{algorithm}

Our dynamic program will maintain an implicit representation of the matrix $C$ from the previous theorem. For extracting the smallest non-zero terms it suffices to look at the numerators of $C$, as long as the denominators are normalized.

\begin{lemma}\label{lemma:coefficient_extraction}
  Suppose $\frac{g(x)}{h(x)} \in \Z(x)$ for some $g(x) = \sum_i c_i x^i$ and normalized $h(x)$, and that $\frac{g(x)}{h(x)} = \sum_i d_i x^i$. Then $i$ is the smallest number such that $c_i \neq 0$ if and only if it is the smallest such number such that $d_i \neq 0$.
\end{lemma}
\begin{proof}
    Suppose $h(x) = \sum_i e_i x^i$ with $e_0 = 1$. Then $c_0 = e_0 d_0$ and hence $c_0 = 0$ if and only if $d_0 = 0$. Further, $c_i = \sum_{j=0}^i e_j d_{i-j}$. Hence, if $c_\ell = d_\ell = 0$ for all $\ell < i$ then $c_i = 0$ if and only if $d_i = 0$.
\end{proof}

Instead of working with quotient approximations directly, the dynamic program will maintain in $\DynFOar$ evaluations of numerators and denominators under several numbers from $a \in \N$. 
By interpolating the polynomials we can extract the smallest non-zero term from this representation in $\FOMaj(+, \times)$, even when the evaluation is done modulo several primes.

\begin{lemma}\label{lemma:extraction}
For all numbers $d, d', e \in \N$ with $d' > d$ there are numbers $e', b \in \N$ such that the following is true.
Fix a domain of size $n$. Suppose $g(x) = \sum_i c_i x^i \in \Z[x]$ with $\deg g(x) \leq n^d$ and $\|g(x) \| \leq 2^{n^e}$. Let $S \subseteq [n^{d'}]_0 \in \N$ with $|S| \geq n^d+1$, and $P$ a set of $n^{e'}$ primes (among the first $n^b$ numbers). The smallest $i$ such that $c_i \neq 0$ can be defined in $\FOMaj(+, \times)$ from a relation that stores the value $g(a) \pmod p$ for each $a \in S$ and each $p \in P$.
\end{lemma}
\begin{proof}
The value $g(a)$ with $a \leq n^{d'}$ is bounded by $2^{n^{e'}}$ for some $e'$ that only depends on $d, d'$ and $e$. As the product of the primes in $P$ exceeds this number, $g(a)$ is uniquely determined by the Chinese remainder representation given by the values $g(a) \pmod p$, and can be decoded in $\FOMaj(+,\times)$ \cite[Theorem 4.1]{HesseAB02}.
The statement follows from the fact that $g(x)$ of degree at most $n^d$ is uniquely determined by the values $g(a)$ for $n^d+1$ pairwise distinct data points $a$, and Cauchy interpolation is in $\FOMaj(+,\times)$ \cite[Corollary 6.5]{HesseAB02}.
\end{proof}

\begin{proofsketchof}{Theorem \ref{theorem:inverse_coefficients_tcquery}}
   Suppose $C \in \Z(x)$ is a normalized quotient $n$-approximation of $A^{-1}$. Then by Lemma \ref{lemma:coefficient_extraction}, the smallest 
   $i < n$ such that the $i$-th coefficient of the $st$-entry of $A^{-1}$ is non-zero is equal to the smallest such $i$ for the numerator of the $st$-entry of $C$. This $i$ can be extracted from the relations stated in Lemma \ref{lemma:extraction} by a $\FOMaj(+, \times)$ formula. 
   
   Our goal is therefore to maintain relations that store the values from Lemma \ref{lemma:extraction} for each entry of $C$. An inspection of the proof of Theorem \ref{theorem:fewerChanges} shows that it suffices to exhibit a dynamic program that maintains such relations for $k \df \frac{\log n}{\log \log n}$ changes of size $k$ each, starting from initial auxiliary relations with respect to a normalized quotient $n$-approximation $C \in \Z(x)$ with numerators of degree at most $n$ and denominator 1. More details of this initialisation are given towards the end of this sketch.
   
   For a domain of size $n$, let $S \df [n^{\lambda}]$ and let $P$ be the set of the first $n^{\mu}$ primes for $\lambda$ and $\mu$ to be determined later. The dynamic program implicitly maintains a quotient $n$-approximation $C \in \Z(x)$ of $A^{-1}$ as follows. For each $a \in S$, each $p \in P$ and each entry $(s, t)$ of $C$, it maintains $g(a) \pmod p$ and $h(a) \pmod p$ if $\frac{g(x)}{h(x)}$ is the $st$-th entry of $C$. Whenever $h(a) = 0 \pmod p$ then $p$ is declared invalid for $a$; whenever $h(a) = 0$ for some denominator $h(x)$ and $a \in S$ then $a$ is declared invalid. The set of primes valid for a value $a$ is denoted by~$P_a$. 
   
   The initial amount of values in $S$ and in $P$ is chosen such that after $k$ changes, sufficiently many valid values remain in $S$ and in each $P_a$ in order to apply Lemma \ref{lemma:extraction} for extracting the smallest non-zero coefficients of denominators from this implicit representation. %

   Suppose $C$ is a quotient $m$-approximation of $A^{-1}$ implicitly stored by the dynamic program. Then we denote by $\calC_a \pmod p$ the evaluation of $C$ at position $a$ modulo prime $p$ for valid $a$ and $p$. By $\calC_a$ we denote the tuple $(\calC_a \pmod {p_1}, \ldots, \calC_a \pmod {p_\eta})$ where $p_1, \ldots, p_{\eta}$ are the primes still valid for $a$. By $\calC$ we denote the tuple $\calC_{a_1}, \ldots, \calC_{a_\kappa}$ where $a_1, \ldots, a_\kappa$ are valid numbers. 

   Formally the program uses a relation $D$ that stores a tuple $(\bar a, \bar p, s, t, \bar v)$ if and only if (i) $a$ is valid, (ii) $p$ is valid for $a$, and (iii) $v$ is the value of the denominator of the $st$-th entry modulo $p$; and similarly a relation $N$ for storing the numerators. These relations encode $\calC_a \pmod p$ as well as the sets $S$ and $P_a$. In the following we abstain from using this formal perspective for the sake of clarity. The descriptions to follow can be easily translated to this formal framework.

   We describe how the program deals with a change $\Delta A$; afterwards we discuss how the auxiliary data is initialized. When a change $\Delta A$ occurs the dynamic program updates the sets $S$ and $P_a$ as well as the tuple $\calC$ according to Algorithm \ref{algorithm:dynfo_updates_implicit}.

  \begin{algorithm}
    \caption{Updating the auxiliary data for distances in $\DynFO(+, \times)$ with $\TC^0$-query}\label{algorithm:dynfo_updates_implicit}
    \begin{algorithmic}[1]
        \Input A change $\Delta A$
        \State \begin{varwidth}[t]{\linewidth}
      Decompose $x \Delta A$ into $UBV$ with $U \in \Z^{n \times k}, B \in \Z[x]^{k \times k},$ and $V \in \Z^{k \times n}$\par
        \hskip\algorithmicindent according to Lemma \ref{lemma:change_decomposition}. 
      \end{varwidth}
        \For{each $a \in S$}
          \For{each prime $p \in P_a$}
            \For{all $(s, t) \in [\frac{\log n}{\log \log n}]^2$}
              \State \begin{varwidth}[t]{\linewidth}
      Compute $g_{st}(a) \pmod p$ and $h_{st}(a) \pmod p$ where $\frac{g_{st}(x)}{h_{st}(x)}$ is the $st$-th entry\par
        \hskip\algorithmicindent of $(I+BVCU)^{-1}$, following Algorithm \ref{algorithm:woodbury}. 
      \end{varwidth}
              \State If $h_{st}(a) = 0 \pmod p$ then remove $p$ from $P_a$
              \State If $P_a = \emptyset$ then remove $a$ from $S$
              \State \begin{varwidth}[t]{\linewidth}
      If $P_a \neq \emptyset$ then compute $\calC'_a \pmod p$ according to $C-CU(I+BVCU)^{-1}BVC$\par
        \hskip\algorithmicindent and following Algorithm  \ref{algorithm:woodbury}. 
      \end{varwidth}
          \EndFor
        \EndFor
      \EndFor
    
    \end{algorithmic}
    \end{algorithm}
  
  All steps can be performed in $\FO(+, \times)$. The loops from Lines 2--4 are executed in parallel. For Line 5, observe that $I+BVCU \pmod p$ can be computed for the same reason as in the proof of Theorem \ref{theorem:non_zero_inverse}.

  Let us analyze the the necessary amount of values in $S$ and $P$. By Lemma \ref{lemma:quotient_approximation_update}(a), the degrees of denominators $h(x)$ of $C$ grow by a factor $k^3$ after each change. Thus, after $k$ change steps the degree of denominators is bounded by $nk^{3k} \in \bigO(n^r)$ for some $r \in \N$. Therefore each denominator $h$ evaluates to $0$ for at most $\bigO(n^r)$ many $a \in S$. All in all there are at most $\bigO(n^{r+2})$ many $a \in S$ such that $h(a) = 0$ for some denominator $h$ of $C$ after one update step, and at most $\bigO(n^{r+3})$ such $a$ in the course of $k$ change steps. 
  
  By Lemma \ref{lemma:quotient_approximation_update}(b), the values of coefficients of $C'$ are bounded by $(\|C\| k \deg C)^{\bigO(k^3)}$ after one change step. Thus, after $k$ changes, the values $h(a)$ for $a \in S$ are bounded by $2^{\bigO(n^{s(r)})}$ for some $s \in \N \to \N$. Thus, if $h(a) \neq 0$ then $h(a) \equiv 0 \pmod p$ for at most $\bigO(n^{s(r)})$ many primes $p$ and for a denominator $h(x)$ of $C$. 
All in all, at most $\bigO(n^{s(r)+2})$ many primes $p$ are removed from $P_a$ in Line 6 in one step. If $h(a) \neq 0$ after each of $k$ many changes, then at most $\bigO(n^{s(r)+3})$ many primes $p$ have been removed from $P_a$. If $h(a) = 0$ after some change, then $h(a) \equiv 0 \pmod p$ for all primes $p$ and therefore $a$ is removed from $S$ in Line 7.
  
Thus, there exists numbers $\lambda, \mu$ such that if one starts with $S = [n^{\lambda}]_0$ and $n^{\mu}$ primes~$P$, then after $k = \frac{\log n}{\log \log n}$ many changes there are still at least $n^r$ numbers $a$ in $S$, each of them with a set $P_a$ of size at least $n^{e'}$, for the number $e'$ that is guaranteed to exists by Lemma \ref{lemma:extraction} applied to $d \df r, d' \df \lambda, e \df s(r)$.  In particular one can define distances with a $\FOMaj(+, \times)$ formula according to Lemma \ref{lemma:extraction}.

We remark shortly on how to initialize the auxiliary data. %
In the proof of Theorem \ref{theorem:fewerChanges}, the initialisation is only applied to the input database. 
Inspecting the proof one can see that the initialisation also has access to the auxiliary relations that are maintained for its input database.
So, a $\FOMaj(+,\times)$ initialisation can interpolate the polynomials $g(x), h(x)$ for each entry of $C$, compute the first $n+1$ coefficients $d_0, \ldots, d_n$ of $\frac{g(x)}{h(x)}$, and therefore obtain a new entry for the quotient $n$-approximation with numerator $\sum_{i=0}^n d_ix^i$ and denominator~$1$. 
Additionally, the initialisation can compute $C(a) \pmod p$ for all $a \in [n^{\lambda}]_0$ and all $n^{\mu}$ primes $p \in P$.
As $\FOMaj(+,\times) = \text{uniform } \TC^0 \subseteq \NC^1 \subseteq \AC[\frac{\log n}{\log \log n}]$ \cite[Theorem 4.3]{ChandraSV84}, by the sketched generalisation of Theorem \ref{theorem:fewerChanges} it is sufficient to maintain the query result for $k =  \frac{\log n}{\log \log n}$ change steps.
\end{proofsketchof}

 }

\end{document}